\documentclass[sigconf,nonacm]{acmart}

\renewcommand\footnotetextcopyrightpermission[1]{}
\usepackage{booktabs}
\usepackage{amsmath}
\usepackage{graphicx}

\newtheorem{assumption}{Assumption}

\newcommand{\lneg}{\lambda_{\mathrm{neg}}}
\newcommand{\Mk}{\mathcal{M}_\kappa}

\begin{document}

\title{Between-User Collapse Under Popularity-Biased Feedback:\\
A Centered-Covariance Theorem and Computable Phase Boundary}

\author{Sahil Medepalli}
\affiliation{%
  \institution{Independent Researcher}
  \city{Flower Mound}\state{Texas}\country{USA}}
\email{sahilmedepalli@gmail.com}

\begin{abstract}
We study how popularity-biased BPR training reshapes the \emph{between-user} geometry of collaborative-filtering embeddings. We work with the mean-centered user covariance $C=\tfrac1n U^\top H U$, the object that measures how distinguishable users are from one another, as opposed to the uncentered second moment used in prior work. We prove that under popularity-biased feedback with stationary items, $C$ converges to a steady state proportional to the item-noise covariance $Q$. Thus between-user spread collapses toward a noise floor. We derive a closed-form, computable phase boundary in the training hyperparameters $(\alpha,\lneg,\gamma,d)$ separating contraction from expansion, and validate both directional predictions on MovieLens-25M. We then examine the limits of the effect. At deployment-scale regularization the predicted contraction is real and policy-driven but small, and it is not reflected in any recommendation-level metric we measured. The $\alpha$-driven anisotropic-collapse mechanism operates only at regularization strengths that degrade the recommender. A deployment-time restoration intervention derived from the theory does not improve recommendation quality. The boundary is computable from a trained model's embeddings, item interaction counts, and training hyperparameters, so a practitioner can check whether a deployed system sits in the strong-collapse regime without simulating the feedback loop. In our experiments the boundary places deployable settings far from that regime.
\end{abstract}


\maketitle

\section{Introduction}

Recommendation feedback loops homogenize users. Training on algorithmically confounded data homogenizes behavior without increasing utility \cite{chaney2018confounding}, and iterated retraining makes popularity bias stronger~\cite{mansoury2020feedback}. These are \emph{behavioral} accounts of a long-term consequence. What is less understood is how popularity-biased training reshapes the geometry of the \emph{user embedding space} over time, and we lack a good test of whether that geometric effect has measurable recommendation-level consequences.

There is a growing line of work that studies embedding collapse in collaborative filtering \cite{peng2025directspec,chen2024ncl,zhang2023dimensionalcollapse,loveland2025matrixrank}, echoing dimensional collapse in contrastive self-supervised learning \cite{jing2022dimensionalcollapse}. All of it operates on the \emph{uncentered} embedding matrix or its second moment $\Sigma=\tfrac1n U^\top U$, which conflates two distinct phenomena. Under popularity bias, user embeddings drift together toward a shared ``popularity direction'', which is a mean-drift effect recently characterized analytically \cite{liu2026rethinking}. A decline in the stable rank of the uncentered matrix \cite{loveland2025matrixrank} is consistent with mean drift alone, meaning users can move \emph{together} without becoming more similar \emph{to one another}. User homogenization means recommendations becoming identical \emph{across users}, not just worse on average, and is a property of the \emph{centered} covariance $C=\Sigma-\bar{\mathbf u}\bar{\mathbf u}^\top$, which no prior work analyzes under feedback.
The phase boundary we derive only needs a trained model's embeddings, item interaction counts, and training hyperparameters, so a practitioner can determine whether a system sits in the strong-collapse regime before worrying about geometric homogenization.

\textbf{Contributions.} (1)~A theorem on the centered (between-user) covariance: under popularity-biased BPR feedback with stationary items, $C$ converges to a steady state \emph{exactly proportional} to the item-noise covariance $Q$. Centering removes the popularity-drift component, and between-user spread collapses toward a noise floor whose anisotropy is that of $Q$. (2)~A closed-form phase boundary in $(\alpha,\lneg,\gamma,d)$, computable from a trained embedding matrix, item counts, and hyperparameters, separating contraction from expansion regimes. We validate both directions on MovieLens-25M. (3)~An empirical scoping showing that at deployable regularization the contraction is real and policy-driven but \emph{not reflected in recommendation metrics} in our setup, that the $\alpha$-driven anisotropy mechanism is inert at deployable settings, and a theory-derived restoration intervention that slightly \emph{hurts} accuracy.

\section{Related Work}

\textbf{Embedding collapse in CF (uncentered, training-time).} DirectSpec \cite{peng2025directspec} balances the embedding spectrum during training (alignment as low-pass, negative sampling as high-pass filter). nCL \cite{chen2024ncl} prevents dimensional collapse via an alignment-plus-compactness (coding-rate) objective. Zhang et al.\ \cite{zhang2023dimensionalcollapse} show graph convolution shrinks the singular space and counteract it with a LogDet regularizer. Loveland et al.\ \cite{loveland2025matrixrank} show alignment lowers the stable rank of the \emph{raw} user/item matrices. That is an uncentered object whose decline is consistent with pure mean drift. Their companion result shows weight decay encodes popularity into embedding \emph{magnitudes} \cite{loveland2025weightdecay}, which is a norm-level effect. None of these analyze the mean-centered between-user covariance, and none derive a contraction/expansion phase boundary in the training hyperparameters.

Liu et al.\ \cite{liu2026rethinking} prove BPR organizes item embeddings along a dominant popularity direction and propose a directional correction (DDC). Their analysis focuses on the mean-drift component, while our result characterizes what happens to between-user spread \emph{after} the mean drifts. The drift term is eliminated exactly in our analysis.

Simulation studies \cite{chaney2018confounding,mansoury2020feedback} establish the behavioral consequence. We provide the representational mechanism and test whether the geometric effect has behavioral consequences at deployable settings.

The distinction that does the work is centering. Every prior analysis of collapse in collaborative filtering studies an uncentered object, and an uncentered spectrum falls when users translate together just as it falls when they converge on each other. The two are different claims about a recommender, and only the second is homogenization.

\section{Setup and Assumptions}

Let $n$ be the number of users, $d$ the embedding dimension, $U\in\mathbb{R}^{n\times d}$ the user-embedding matrix, and $H=I_n-\tfrac1n\mathbf 1_n\mathbf 1_n^\top$ the centering matrix ($H=H^\top$, $H^2=H$, $H\mathbf 1_n=\mathbf 0$). The centered user-embedding covariance is
\[
C=\tfrac1n U^\top H U=\Sigma-\bar{\mathbf u}\bar{\mathbf u}^\top,
\qquad \Sigma=\tfrac1n U^\top U,\quad
\bar{\mathbf u}=\tfrac1n U^\top\mathbf 1_n,
\]
and the volume measure is $V_C=\log\det(C+\varepsilon I)$, which we track over rounds $t=0,1,2,\dots$ $C$ measures \emph{between-user spread}, and the uncentered $\Sigma=C+\bar{\mathbf u}\bar{\mathbf u}^\top$ conflates it with the squared mean norm. A decrease in $V_C$ means that users are converging toward one another after mean subtraction, which is what user homogenization is.

\begin{assumption}[Popularity-Biased Interaction Distribution for Mainstream Users]
\label{ass:popularity}
Let $p_u(i) = c_{ui} / \sum_k c_{uk}$ be user $u$'s empirical interaction distribution, where $c_{ui}$ is the number of times user $u$ has interacted with item $i$, and let $c_i = \sum_u c_{ui}$. We define the popularity-weighted global distribution, with exponent $\alpha \geq 0$, as:
\[
p_{\mathrm{global}}(i;\alpha) =
\frac{c_i^\alpha}{\sum_{k=1}^m c_k^\alpha},
\]
uniform at $\alpha = 0$ and concentrating on globally popular items as $\alpha$ grows. Define the \emph{mainstream user subpopulation} at threshold $\kappa > 0$:
\[
\Mk = \left\{u :
D_{KL}\!\left(p_u \,\|\,
p_{\mathrm{global}}(\cdot\,;1)\right) \leq
\kappa\right\}.
\]
We restrict attention to users $u \in \Mk$ and approximate their positive sampling distribution as the global popularity distribution, $p(i \mid u) \approx p_{\mathrm{global}}(i;\alpha)$. The approximation error in the expected positive item embedding is bounded by Pinsker's inequality:
\[
\left\|\mathbb{E}_{i \sim p_u}[\mathbf{v}_i]
- \boldsymbol{\mu}_+(\alpha)\right\|
\leq \sqrt{2\kappa} \cdot \max_i \|\mathbf{v}_i\|,
\]
where $\boldsymbol{\mu}_+(\alpha) =
\sum_i p_{\mathrm{global}}(i;\alpha)\mathbf{v}_i$. The bound is small relative to $\|\boldsymbol{\Phi}(\alpha)\|$ whenever the popularity direction is strong.
\end{assumption}

The threshold $\kappa$ trades coverage against approximation error, and the Pinsker bound is what makes the trade explicit. Smaller $\kappa$ tightens the bound and shrinks $\Mk$. We use $\kappa=3.0$. On MovieLens-25M, among the $1000$ most active users the mean $D_{KL}$ is $1.86$ and $95.4\%$ satisfy $D_{KL}\le 3.0$, while among $1000$ users sampled at random the mean is $3.28$ and $41.7\%$ satisfy it. The theorem therefore applies to most active users in the dataset and to a minority of users sampled at random. This is the honest scope. Users whose taste never resembled the popular distribution are not the users a popularity-driven collapse argument is about.

In experiments we manipulate $\alpha$ by popularity-weighted positive resampling, drawing positives from $p_u^{\alpha}(i)\propto c_{ui}^{\alpha}$ rather than uniformly from the user's history. For users in $\Mk$ this is a clean first-order manipulation of the effective global exponent, since $p_u^\alpha \approx p_{\mathrm{global}}^\alpha$ by the same argument that justifies the assumption. We verify the realized exponent after each epoch by an OLS fit of $\log \hat p(i)$ on $\log c_i$, and we accept an arm only when the fitted exponent is within $0.1$ of the target with $R^2>0.9$. Conditions $\alpha\in\{1.5,2.0\}$ exceed the empirical $\hat\alpha=1.11$ on this dataset, and we treat them as stress tests rather than as estimates of realistic system behavior.

\begin{assumption}[BPR Gradient Update]
\label{ass:bpr}
User embeddings are updated via gradient descent on the Bayesian Personalized Ranking loss \cite{rendle2009bpr} with learning rate $\eta > 0$ and $\ell_2$ regularization strength $\gamma > 0$:
\[
\mathbf{u}^{t+1} = (1 - \eta\gamma)\mathbf{u}^t
+ \eta \delta_{ij}(\mathbf{v}_i - \mathbf{v}_j),
\]
where $\delta_{ij} = 1 - \sigma(\mathbf{u}^\top
\mathbf{v}_i - \mathbf{u}^\top \mathbf{v}_j)
\in (0,1)$ and $\sigma$ denotes the sigmoid function. This assumption describes SGD with constant learning rate. Production systems often use adaptive optimizers, and the qualitative mechanism holds under them, though the quantitative phase boundary may shift. We validate the boundary under SGD in simulation
(Appendix~\ref{app:sim}) and the steady-state structure under Adam on
real data (Section~\ref{sec:empirical}).
\end{assumption}

\begin{assumption}[Independent Negative Sampling]
\label{ass:negsampling}
Negative items are sampled uniformly at random from the item catalog with $\lneg$ negative samples per positive sample. Training pairs $(i_u, j_u)$ are sampled independently across users within each step.
\end{assumption}

\begin{assumption}[Stationary Item Embeddings]
\label{ass:stationary}
Item embeddings $\{\mathbf v_i\}_{i=1}^m$ are fixed during the analysis period, which isolates the user-embedding dynamics. We validate the theorem directly in this regime by training with the item embeddings frozen (Section~\ref{sec:empirical}), and we do not appeal to late-training convergence. When item embeddings co-evolve in the full feedback loop, item dynamics introduce an additional expansive force that we characterize empirically as out-of-scope behavior. The stationary-item theorem describes the user-side mechanism in isolation.
\end{assumption}

\begin{assumption}[Popularity Non-Degeneracy]
\label{ass:nondegen}
The item embedding distribution is popularity-correlated and non-degenerate. Formally, letting
\[
\boldsymbol{\Phi}(\alpha) = \sum_{i=1}^m
\left(\frac{c_i^\alpha}{\sum_k c_k^\alpha}
- \frac{1}{m}\right)\mathbf{v}_i,
\]
we require $\boldsymbol{\Phi}(\alpha) \neq \mathbf{0}$ for $\alpha > 0$, so that the mean-drift direction is nontrivial, and $Q(\alpha,\lneg) \succ 0$, so that the item-noise covariance appearing in Theorem~\ref{thm:contraction} is positive definite. We verify this empirically, where $\log\det Q$ is finite.
\end{assumption}

Additionally, we require $0 < \eta\gamma < 2$ to ensure convergence of the mean user embedding recurrence. Throughout we assume $C$ is positive definite, which holds generically when $n \gg d$. MovieLens-25M satisfies this with $n = 162{,}541$ and $d = 64$. The $\varepsilon I$ ridge provides numerical stability and is not required theoretically.

\emph{The $\bar\delta$ approximation.} Throughout the analysis we replace $\delta_{ij}(\mathbf{u})$ with its expectation $\bar{\delta}$, treating it as constant across users and samples \cite{bottou2018optimization}. Since $\delta_{ij} \in (0,1)$, Popoviciu's inequality gives $\mathrm{Var}(\delta_{ij}) \leq \tfrac{1}{4}$, and by Cauchy--Schwarz, $|\mathrm{Cov}(\delta_{ij},\mathbf{v}_i)| \leq \tfrac{1}{2}\sqrt{\mathrm{Var}(\mathbf{v}_i)}$. The direction of the approximation error depends on the regime. When
regularization pins scores near zero, $\bar\delta$ stays near $0.5$
and the approximation is nearly exact. When the model fits,
$\bar\delta$ decays and its state dependence adds a contractive force
the approximation misses (Appendix~\ref{app:sim}). As training progresses and embeddings converge toward $\boldsymbol{\Phi}(\alpha)$, $\delta_{ij}$ becomes smaller for popular positives, introducing a natural self-limiting mechanism, which is exactly the decay we observe empirically ($\bar\delta = 0.371$ at warmup, plateauing to $0.186$).

\section{Main Result}

\begin{theorem}[Collapse of Between-User Spread under Popularity-Biased Feedback]
\label{thm:contraction}
Under Assumptions \ref{ass:popularity}--\ref{ass:nondegen}, for users in $\Mk$, with $0<\eta\gamma<2$ and the $\bar\delta$ approximation, the centered user covariance $C^t=\tfrac1n U^{t\top}HU^t$ converges to the steady state determined by $\mathbb E[\Delta C]=0$:
\[
C^{*}=\frac{\eta\bar\delta^{2}}{2\gamma}\,Q(\alpha,\lneg)
=\frac{\bar\delta^{2}\tau}{\gamma^{2}}\,Q,\qquad \tau=\tfrac{\eta\gamma}{2}.
\]
Thus $C^{*}$ is proportional to the item-noise covariance $Q$, meaning it shares $Q$'s eigenvectors with eigenvalues $\lambda_k^{C,*}=(\bar\delta^{2}\tau/\gamma^{2})\mu_k$, where $\mu_k$ is the $k$-th eigenvalue of
\[
Q(\alpha,\lneg)=
S_+(\alpha)-\boldsymbol\mu_+(\alpha)\boldsymbol\mu_+(\alpha)^\top
+\tfrac{1}{\lneg}\bigl(S_--\boldsymbol\mu_-\boldsymbol\mu_-^\top\bigr),
\]
$S_+(\alpha)=\sum_i w_i(\alpha)\mathbf v_i\mathbf v_i^\top$, $w_i(\alpha)=c_i^\alpha/\sum_k c_k^\alpha$, $S_-=\tfrac1m\sum_j\mathbf v_j\mathbf v_j^\top$. Between-user spread collapses toward the noise floor $\mu_k$ scaled by $\bar\delta^2\tau/\gamma^2$, with no dependence on the popularity direction $\boldsymbol\Phi(\alpha)$, which cancels exactly under $H\mathbf 1_n=\mathbf 0$. The steady-state anisotropy is inherited from $Q$. Finally $\lambda_k^{C,*}>0$ for all $k$ under Assumption~\ref{ass:nondegen} and $\lneg<\infty$, so $C$ collapses toward, but not to, zero.
\end{theorem}

\begin{proof}
We write $a = 1-\eta\gamma$ throughout and give the four steps in full.

\smallskip
\noindent\textbf{Step 1 (per-user update).}
Fix a user $u\in\Mk$ and a training step. The user-side BPR gradient for the triple $(u, i_u, j_u)$ is $-\delta_{i_uj_u}(\mathbf v_{i_u}-\bar{\mathbf v}_{-,u})$, where $\bar{\mathbf v}_{-,u}$ is the mean of the $\lneg$ sampled negatives. With $\ell_2$ strength $\gamma$ and the $\bar\delta$ approximation, Assumption~\ref{ass:bpr} gives
\[
\mathbf u_u^{t+1}=a\,\mathbf u_u^t+\eta\,\bar\delta\bigl(\mathbf v_{i_u}-\bar{\mathbf v}_{-,u}\bigr).
\]
Under Assumptions~\ref{ass:popularity} and \ref{ass:negsampling}, $\mathbb E[\mathbf v_{i_u}]=\boldsymbol\mu_+(\alpha)$ and $\mathbb E[\bar{\mathbf v}_{-,u}]=\boldsymbol\mu_-$, so the expected step is $\eta\bar\delta\bigl(\boldsymbol\mu_+(\alpha)-\boldsymbol\mu_-\bigr)=\eta\bar\delta\,\boldsymbol\Phi(\alpha)$, recovering the $\boldsymbol\Phi$ of Assumption~\ref{ass:nondegen}. Splitting the step into mean and fluctuation,
\[
\mathbf u_u^{t+1}=a\,\mathbf u_u^t+\eta\bar\delta\,\boldsymbol\Phi(\alpha)+\eta\,\boldsymbol\epsilon_u,
\qquad \mathbb E[\boldsymbol\epsilon_u]=\mathbf 0 .
\]
The fluctuation has covariance $\mathbb E[\boldsymbol\epsilon_u\boldsymbol\epsilon_u^\top]=\bar\delta^2 Q(\alpha,\lneg)$, where the positive draw contributes $S_+(\alpha)-\boldsymbol\mu_+\boldsymbol\mu_+^\top$ and the mean of $\lneg$ independent uniform negatives contributes $\tfrac1{\lneg}(S_--\boldsymbol\mu_-\boldsymbol\mu_-^\top)$. Stacking users into rows,
\[
U^{t+1}=aU^{t}+\eta\bar\delta\,\mathbf 1_n\boldsymbol\Phi(\alpha)^\top+\eta E,
\]
where row $u$ of $E$ is $\boldsymbol\epsilon_u$ and the rows are independent across users by Assumption~\ref{ass:negsampling}.

\smallskip
\noindent\textbf{Step 2 (centering kills the drift).}
The drift term is rank one with every row equal, so left-multiplying by $H$ annihilates it exactly:
\[
H\bigl(\mathbf 1_n\boldsymbol\Phi(\alpha)^\top\bigr)=(H\mathbf 1_n)\boldsymbol\Phi(\alpha)^\top=\mathbf 0 .
\]
This is an algebraic identity, not a small-term approximation, and it holds for every $\alpha$, every $\bar\delta$ and every $t$. Hence
\[
HU^{t+1}=aHU^{t}+\eta HE .
\]
Using $H^\top=H$ and $H^2=H$,
\[
C^{t+1}=\tfrac1n\bigl(HU^{t+1}\bigr)^{\!\top}\bigl(HU^{t+1}\bigr)
=a^2C^{t}+\tfrac{a\eta}{n}\bigl(U^{t\top}HE+E^\top HU^{t}\bigr)+\tfrac{\eta^2}{n}E^\top HE .
\]
The cross terms vanish in expectation because the step-$t$ sampling is independent of $U^t$ and $\mathbb E[E]=0$. For the last term, $E^\top HE=E^\top E-n\bar{\boldsymbol\epsilon}\bar{\boldsymbol\epsilon}^\top$ with $\bar{\boldsymbol\epsilon}=\tfrac1nE^\top\mathbf 1_n$, and since the rows are independent with covariance $\bar\delta^2Q$,
\[
\mathbb E\bigl[\tfrac1nE^\top HE\bigr]=\bar\delta^{2}Q-\tfrac1n\bar\delta^{2}Q=\bigl(1-\tfrac1n\bigr)\bar\delta^{2}Q .
\]
The exact recursion is therefore
\[
\mathbb E[C^{t+1}]=a^{2}C^{t}+\eta^{2}\bigl(1-\tfrac1n\bigr)\bar\delta^{2}Q,
\]
that is, $\mathbb E[\Delta C]=(-2\eta\gamma+\eta^{2}\gamma^{2})C^{t}+\eta^{2}(1-\tfrac1n)\bar\delta^{2}Q$. Dropping the $O(\eta^2\gamma^2)$ term and the $O(1/n)$ finite-sample correction, both negligible in our setting where $\eta\gamma\le10^{-4}$ and $n\gg d$, gives the first-order form
\[
\mathbb E[\Delta C]=-2\eta\gamma\,C^t+\eta^2\bar\delta^2\,Q(\alpha,\lneg).
\]
Only the noise term survives centering. The popularity direction has left the covariance dynamics entirely.

\smallskip
\noindent\textbf{Step 3 (where the drift went).}
Averaging the Step 1 update over users gives $\mathbb E[\bar{\mathbf u}^{t+1}]=a\,\mathbb E[\bar{\mathbf u}^{t}]+\eta\bar\delta\boldsymbol\Phi(\alpha)$, a linear recurrence that converges if and only if $|a|<1$, which is the condition $0<\eta\gamma<2$. Its fixed point is
\[
\bar{\mathbf u}^{\infty}=\frac{\eta\bar\delta}{1-a}\boldsymbol\Phi(\alpha)=\frac{\bar\delta}{\gamma}\boldsymbol\Phi(\alpha).
\]
The mean does move along the popularity direction, and it moves by an amount that grows as regularization weakens. This is the drift component characterized by \cite{liu2026rethinking}, and by Step 2 it doesn't enter the centered dynamics. An analysis of the uncentered $\Sigma=C+\bar{\mathbf u}\bar{\mathbf u}^\top$ sees $\|\bar{\mathbf u}^\infty\|^2=(\bar\delta/\gamma)^2\|\boldsymbol\Phi(\alpha)\|^2$ added on top, which is why an uncentered spectrum responds to popularity bias even when between-user spread does not.

\smallskip
\noindent\textbf{Step 4 (solving the matrix fixed point).}
Setting $\mathbb E[\Delta C]=0$ in the first-order form gives $2\eta\gamma\,C^{*}=\eta^{2}\bar\delta^{2}Q$, so
\[
C^{*}=\frac{\eta\bar\delta^{2}}{2\gamma}\,Q(\alpha,\lneg)=\frac{\bar\delta^{2}\tau}{\gamma^{2}}\,Q,\qquad \tau=\frac{\eta\gamma}{2}.
\]
This solves the matrix equation directly, so it holds independently of the coordinate system and requires no eigenbasis argument. Because $C^{*}$ is a positive scalar multiple of $Q$, the two are simultaneously diagonalizable, $\lambda_k^{C,*}=(\bar\delta^{2}\tau/\gamma^{2})\mu_k$, and
\[
V_C^{*}=\sum_{k=1}^{d}\log\lambda_k^{C,*}=d\log\!\bigl(\bar\delta^{2}\tau/\gamma^{2}\bigr)+\log\det Q .
\]
Since $Q\succ 0$ by Assumption~\ref{ass:nondegen} and $\lneg<\infty$ keeps the $\tfrac1{\lneg}$ term from vanishing, every $\mu_k>0$ and every $\lambda_k^{C,*}>0$.

Convergence follows from the recursion itself. The map $C\mapsto a^{2}C+\eta^{2}(1-\tfrac1n)\bar\delta^{2}Q$ of Step 2 is a contraction with rate $a^{2}<1$ whenever $0<\eta\gamma<2$, so the iteration reaches its unique fixed point from any positive definite start. That fixed point is
\[
\frac{\eta\bar\delta^{2}}{2\gamma}\cdot\frac{1-1/n}{1-\eta\gamma/2}\;Q,
\]
which is $C^{*}$ up to the same two factors dropped above. At our settings the correction is under $0.01$ nats of $V_C$, against effects of hundreds of nats, so we work with $C^{*}$ throughout.

Tracking eigenvalues in the frozen $C^{0}$ eigenbasis instead of solving the matrix equation gives $\lambda_k^{*}=(\eta\bar\delta^2/2\gamma)\,\mathbf q_k^\top Q\mathbf q_k$, and $\sum_k\log(\mathbf q_k^\top Q\mathbf q_k)\ge\log\det Q$ by Hadamard's inequality, with equality only when $Q$ is diagonal in that basis. The frozen-basis route therefore yields an upper bound on the steady-state volume rather than the value. The exact fixed point uses $Q$'s own spectrum.
\end{proof}

\begin{corollary}[Computable Phase Boundary]
\label{cor:phaseboundary}
Under the assumptions of Theorem~\ref{thm:contraction}, for $\Mk$, the contraction condition $V_C^{*}<V_C^{0}$ is equivalent to
\[
\log\det Q(\alpha,\lneg)
<\sum_{k=1}^d\log\tilde\lambda_k^{C,0},\qquad
\tilde\lambda_k^{C,0}=\frac{2\gamma\,\lambda_k^{C,0}}{\bar\delta^2\eta},
\]
and equivalently $\log\det Q<\log\det C^0+d\log(2\gamma/\bar\delta^2\eta)$. All quantities are computable from the trained embedding matrix restricted to $\Mk$, item interaction counts, and training hyperparameters.
\end{corollary}

\begin{proof}
By Theorem~\ref{thm:contraction}, $V_C^{*}=d\log(\bar\delta^{2}\tau/\gamma^{2})+\log\det Q$ with $\tau=\eta\gamma/2$, and $V_C^{0}=\sum_k\log\lambda_k^{C,0}$. The condition $V_C^{*}<V_C^{0}$ rearranges to
\[
\log\det Q<\sum_{k=1}^{d}\log\lambda_k^{C,0}-d\log\!\bigl(\bar\delta^{2}\tau/\gamma^{2}\bigr)
=\sum_{k=1}^{d}\log\frac{\gamma^{2}\lambda_k^{C,0}}{\bar\delta^{2}\tau},
\]
and $\gamma^{2}/\bar\delta^{2}\tau=2\gamma/\bar\delta^{2}\eta$.
\end{proof}

\emph{Remarks.} (i)~$\mathrm{erank}(C^*)=\mathrm{erank}(Q)$. As $\alpha\to\infty$, $S_+$ concentrates on popular items, lowering $\mathrm{erank}(Q)$ and hence $\mathrm{erank}(C^*)$. Note that the anisotropy mechanism runs through $Q$, not through alignment with $\boldsymbol\Phi$. (ii)~$\log\det Q$ \emph{decreases} with $\lneg$. More negative sampling shrinks the $1/\lneg$ term, lowers the noise floor, and makes contraction \emph{easier}, whereas fewer negatives make it harder. (iii)~The boundary depends on $\gamma$ through $d\log(2\gamma/\bar\delta^2\eta)$, which moves by $d\log 10$ per decade of $\gamma$. At $d=64$ that is about $147$ nats per decade, and it is why the boundary in our experiments is governed by regularization strength rather than by $\alpha$.

\subsection{Evaluating the boundary on a trained model}
\label{sec:worked}

Corollary~\ref{cor:phaseboundary} is meant to be run, so we state the procedure. Given a trained checkpoint, the item interaction counts, and the training configuration:

\begin{enumerate}
\item Compute $D_{KL}(p_u\|p_{\mathrm{global}}(\cdot\,;1))$ for each user and keep those with $D_{KL}\le\kappa$. This gives $\Mk$ and is cached once.
\item Form $C^{0}=\tfrac1{|\Mk|}U_{\Mk}^\top HU_{\Mk}$ on the retained rows and take $\log\det(C^0+\varepsilon I)$. At $d=64$ this is a Cholesky factorization and runs in milliseconds.
\item Fit $\alpha$ from the interaction counts, form the weights $w_i(\alpha)=c_i^\alpha/\sum_k c_k^\alpha$, and build $Q(\alpha,\lneg)$ from the item embeddings by the expression in Theorem~\ref{thm:contraction}. Both $S_+$ and $S_-$ are $d\times d$ accumulations over the catalog.
\item Estimate $\bar\delta$ as the mean of $\delta_{ij}$ over a representative training batch at the same checkpoint.
\item Compare $\log\det Q$ against $\log\det C^0+d\log(2\gamma/\bar\delta^2\eta)$. The system is in the contraction regime when the left side is smaller, and the difference is the margin in nats.
\end{enumerate}

On MovieLens-25M with $d=64$, $\eta=0.001$, $\lneg=1$ and $\bar\delta=0.371$ measured at warmup, this gives $\log\det Q=-405.77$ against a right-hand side of $-30.91$ at $\gamma=0.1$ and $-472.89$ at $\gamma=10^{-4}$. The margins are $+374.87$ and $-67.12$. Both $C^{0}$ and $\bar\delta$ are read from the same checkpoint, so the boundary is a statement about the state it is evaluated at. Section~\ref{sec:magnitude} predicts the plateau instead, and uses the plateau's own $\bar\delta=0.186$. The whole computation is a handful of $d\times d$ operations and one pass over the catalog. It requires no simulation of the feedback loop, meaning a practitioner can see whether the mechanism is even live in their system before spending any time on measuring it.

\section{Empirical Findings}
\label{sec:empirical}

\textbf{Setup.} All experiments use MovieLens-25M with $d=64$, $n=162{,}541$ users, and $m=59{,}047$ items. The mainstream set $\Mk$ with $\kappa=3.0$ gives $63{,}854$ users, and the fitted popularity exponent is $\hat\alpha=1.11$ ($R^2=0.953$). We use $\lneg=1$, $\eta=0.001$, the Adam optimizer, and an $\varepsilon=10^{-6}$ ridge throughout.

\subsection{Phase boundary validation}
\label{sec:5.1}
With items frozen, Assumption~\ref{ass:stationary} is satisfied by construction. We measure $\log\det Q=-405.77$, and the corollary's right-hand side evaluates to $-30.91$ at $\gamma=0.1$ and $-472.89$ at $\gamma=10^{-4}$, predicting \textbf{contraction} at $\gamma=0.1$ (margin $+374.87$) and \textbf{expansion} at $\gamma=10^{-4}$ (margin $-67.12$). Frozen-item training shows that $V_C$ contracts at $\gamma=0.1$ ($-496.82\to-592.44$, $-19.25\%$) and expands at $\gamma=10^{-4}$ ($+28.55\%$, Figure~\ref{fig:boundary}). The two settings sit on opposite sides of the computed boundary and the observed sign matches the predicted sign in each.

\begin{figure}[t]
\centering
\includegraphics[width=0.82\linewidth]{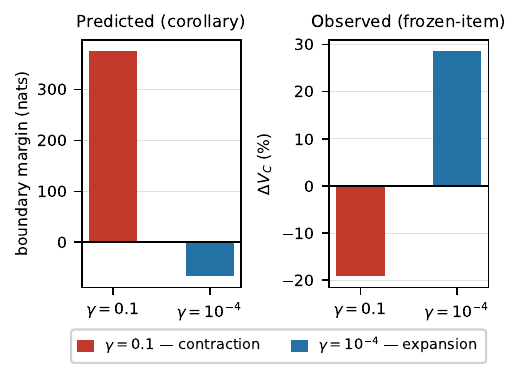}
\caption{Phase-boundary validation with frozen items: at $\gamma=0.1$ the corollary predicts contraction (margin $+374.87$) and $V_C$ contracts by $19.25\%$. At $\gamma=10^{-4}$ it predicts expansion (margin $-67.12$) and $V_C$ expands by $+28.55\%$. The predicted sign matches the observed sign on both sides of the boundary.}
\label{fig:boundary}
\end{figure}

\subsection{Structure and magnitude under Adam}
\label{sec:magnitude}
Regressing $\log\lambda_k(C^*)$ on $\log\mu_k(Q)$ at the $\gamma=0.1$ plateau yields slope $0.954$, $R^2=0.990$ (Figure~\ref{fig:slope}), robust to $\varepsilon$-tail removal ($0$ of $64$ eigenvalues in the $\varepsilon$ tail). This indicates the $Q^1$ law of Theorem~\ref{thm:contraction} persists under Adam (rather than a $Q^{1/2}$ law). The deviation is not specific to Adam. In a fully synthetic SGD environment satisfying every
assumption, the same regression yields slope $0.926$ with $R^{2}$ of
$0.998$ (Appendix~\ref{app:sim}), so the compression appears under the
theorem's own optimizer. We attribute it to the state dependence of
$\delta_{ij}$, which the $\bar\delta$ approximation discards. Learning happens the fastest along high-variance item directions, which suppresses
noise injection exactly where $\mu_k$ is largest.

\begin{figure}[t]
\centering
\includegraphics[width=0.85\linewidth]{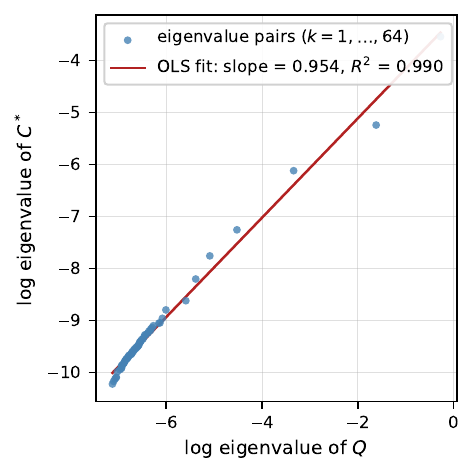}
\caption{Steady-state structure under Adam at the $\gamma=0.1$ plateau with frozen items ($\Mk$ users): $\log\lambda_k(C^*)$ against $\log\mu_k(Q)$ across all $d=64$ eigenpairs, with OLS fit (slope $0.954$, $R^2=0.990$), confirming $C^*\propto Q$. The intercept offset is the Adam scale shift discussed in Section~\ref{sec:magnitude}. This is the strong-regularization regime, not a deployable setting.}
\label{fig:slope}
\end{figure}

The magnitude, however, is off, and the regression characterizes the offset exactly. At the plateau $\bar\delta=0.186$, the SGD theory predicts intercept $\log(\eta\bar\delta^{2}/2\gamma)=-8.66$ with unit slope, giving a predicted plateau $V_C^{*}=d\,(-8.66)+\log\det Q\approx-960.2$ against the observed $-592.44$, a ${\approx}368$-nat gap ($\approx 5.75$ nats per dimension).
A constant per-coordinate intercept shift of $+5.46$ ($349.2$ nats total, measured $-3.21$ vs.\ theoretical $-8.66$) plus the mild spectral compression of the sub-unit slope ($(0.954-1)\log\det Q\approx+18.5$ nats) decompose the gap exactly, with negligible residual. We can rule out two explanations. Attributing the intercept shift to the $\bar\delta$ approximation would require $\bar\delta_{\mathrm{eff}}\approx 2.86$, impossible for a probability. Non-convergence is also excluded, with $V_C$ being stable within $0.17$ nats over the final $40$ epochs. In the same synthetic environment at $\gamma=0.1$, the SGD prediction lands within $1.7$ nats of
the observed plateau, about $0.11$ nats per dimension against $5.75$ on
real data under Adam (Appendix~\ref{app:sim}). We attribute the offset to Adam's adaptive-normalization scale shift. Direction and structure are predicted, and absolute magnitude requires an Adam-aware derivation (future work).

\subsection{Behavior at deployable $\gamma$}
\label{sec:5.3}
In the full feedback loop at $\gamma=10^{-4}$, $V_C$ declines $-1.75\%$ over $50$ rounds with signal-to-noise $\approx 230$ across seeds, while a null arm (random recommendations) \emph{expands} $+4.50\%$, matching the corollary's stationary-item expansion prediction at this $\gamma$ and establishing the contraction as policy-driven. However, no recommendation-level metric moves in lockstep with $V_C$. Inter-user list similarity falls from ${\sim}0.108$ to ${\sim}0.07$--$0.08$ by rounds 5--10 and flattens (the \emph{wrong} direction for homogenization), and per-user churn rises to ${\sim}0.995$ during rounds 0--10 while $V_C$ is still expanding (drift, not convergence). Additionally, coverage is flat ($0.0119\to{\sim}0.011$--$0.012$) and tail exposure floors at $0.0000$ by round 3, \emph{before} the $V_C$ arc begins, so the setup has limited power to detect tail-side consequences. Although the contraction is detectable, we don't observe corresponding changes in recommendation metrics.

\subsection{The $\alpha$-driven anisotropy mechanism is negligible at deployable $\gamma$}
\label{sec:5.4}
Sweeping $\alpha\in[0,2]$ at $\gamma=10^{-4}$, effective rank is flat ($10.89$ at $\alpha=0$ vs.\ $10.87$ at $\alpha=2$) with healthy accuracy (NDCG@10 ${\sim}0.030$) throughout, and the popularity-alignment share of variance slightly \emph{decreases} with $\alpha$. The $\alpha$-driven anisotropic collapse takes place only at model-breaking $\gamma=0.1$. At deployable $\gamma$ the phase boundary is dominated by the $\ell_2$/prefactor term. (An earlier draft's mechanism attributing anisotropy to alignment with $\boldsymbol\Phi(\alpha)$ rested on a variance-vs.-squared-mean error that centering removes, and these results are inconsistent with that interpretation.)

\subsection{Evaluation of a theory-derived restoration intervention}
\label{sec:5.5}
The theorem suggests a deployment-time restoration of collapsed directions and, at training time, a centered log-det regularizer with gradient $\tfrac2n HU(C+\varepsilon I)^{-1}$ that provably raises the steady-state $V_C$, but unlike post-hoc stabilization for basis consistency \cite{zielnicki2025orthogonal}, the goal here is re-inflating collapsed directions. Empirically, the restoration arm yields NDCG@10 $=0.050$ vs.\ baseline $0.056$, so it slightly hurts, with TailNDCG $=0.0000$ for both arms (measurement floor, disclosed). This is consistent with Section~\ref{sec:5.3}. If between-user contraction has no measurable recommendation-level effects at deployable settings, reversing it shouldn't improve quality, and empirically it doesn't.

\subsection{The boundary in a fully synthetic environment}
We also test the boundary where every assumption holds by construction, with frozen anisotropic items, plain SGD, and all users mainstream (Appendix~\ref{app:sim}). Above the predicted crossover $\gamma^{*} \approx 2.3\times10^{-3}$ the predicted sign matches at all four grid points, and at $\gamma=0.1$ the predicted plateau is
accurate to $1.7$ nats. Below the crossover the theory predicts expansion, but we observe a persistent residual contraction of $8.6$
to $14.6$ percent at one million steps, with $\bar\delta$ decaying $20$ to $23$ percent against about $3$ percent above the boundary. The state dependence of $\delta_{ij}$ therefore adds a contractive force in the fitting regime, so the expansion-side prediction is condition dependent, while the strong-collapse side is reliable in
both simulation and on real data.

\section{Limitations}
Strong contraction manifests only at regularization strengths ($\gamma=0.1$) that degrade the recommender. At deployable $\gamma$ it is small and was not reflected in any recommendation-level metric we measured (tail-exposure floor caveat noted). The $\alpha$-driven anisotropy mechanism does not operate at deployable $\gamma$. Under Adam the steady-state magnitude deviates by a constant factor, and in simulation the expansion-side prediction fails when the model fits, since the state dependence of $\delta_{ij}$ adds a contractive force the $\bar\delta$ approximation misses (Appendix~\ref{app:sim}). The theory-motivated intervention did not improve quality. The scope is mainstream users $\Mk$ and stationary items, validated by construction rather than in deployment.

\subsection{Single dataset}
Every number here comes from MovieLens-25M. The dataset has a public-benchmark popularity profile with fitted $\hat\alpha=1.11$ and $R^2=0.953$, which is a reasonable stand-in for a skewed catalog, but representativeness beyond it is an assumption and not a result. Two quantities in the boundary are dataset-dependent in ways we have not measured elsewhere. The first is $\log\det Q$, which depends on how the item embedding geometry lines up with the popularity weights, and a catalog with a different tail shape moves it. The second is the coverage of $\Mk$ at a given $\kappa$, which was $95.4\%$ of active users here and would be smaller on a platform with more niche consumption. Neither affects the derivation, and both affect where a given system sits relative to the boundary. Cross-dataset validation on the centered object is future work. We would rather say that plainly than generalize from one dataset.

\section{Conclusion}
We give a centered-covariance theorem with exact drift cancellation, and steady state $C^*\propto Q$, with a phase boundary computable from a trained model's embeddings, item counts, and hyperparameters. We also see where the mechanism applies and where it doesn't. The negative results are what make the positive result usable. In our experiments the boundary says deployable settings sit far from the collapse regime, which is worth knowing before spending anything on measuring it.

\bibliographystyle{ACM-Reference-Format}
\bibliography{references}

@inproceedings{loveland2025matrixrank,
  author    = {Loveland, Donald and Wu, Xinyi and Zhao, Tong and Koutra, Danai and Shah, Neil and Ju, Mingxuan},
  title     = {Understanding and Scaling Collaborative Filtering Optimization from the Perspective of Matrix Rank},
  booktitle = {Proceedings of the ACM Web Conference 2025 (WWW '25)},
  year      = {2025},
  pages     = {436--449},
  publisher = {ACM},
  address   = {Sydney, NSW, Australia},
  doi       = {10.1145/3696410.3714904},
  note      = {arXiv:2410.23300}
}

@article{peng2025directspec,
  author  = {Peng, Shaowen and Sugiyama, Kazunari and Liu, Xin and Mine, Tsunenori},
  title   = {Balancing Embedding Spectrum for Recommendation},
  journal = {ACM Transactions on Recommender Systems},
  year    = {2025},
  doi     = {10.1145/3718488},
  note    = {arXiv:2406.12032}
}

@inproceedings{zhang2023dimensionalcollapse,
  author    = {Zhang, Yifei and Zhu, Hao and Chen, Yankai and Song, Zixing and Koniusz, Piotr and King, Irwin},
  title     = {Mitigating the Popularity Bias of Graph Collaborative Filtering: A Dimensional Collapse Perspective},
  booktitle = {Advances in Neural Information Processing Systems 36 (NeurIPS 2023)},
  year      = {2023},
  pages     = {67533--67550}
}

@inproceedings{liu2026rethinking,
  author    = {Liu, Lingfeng and Song, Yixin and Shen, Dazhong and Yin, Bing and Li, Hao and Zhang, Yanyong and Wang, Chao},
  title     = {Rethinking Popularity Bias in Collaborative Filtering via Analytical Vector Decomposition},
  booktitle = {Proceedings of the 32nd ACM SIGKDD Conference on Knowledge Discovery and Data Mining V.1 (KDD '26)},
  year      = {2026},
  pages     = {879--890},
  publisher = {ACM},
  address   = {Jeju Island, Republic of Korea},
  doi       = {10.1145/3770854.3780295},
  note      = {arXiv:2512.10688}
}

@inproceedings{chaney2018confounding,
  author    = {Chaney, Allison J.B. and Stewart, Brandon M. and Engelhardt, Barbara E.},
  title     = {How Algorithmic Confounding in Recommendation Systems Increases Homogeneity and Decreases Utility},
  booktitle = {Proceedings of the 12th ACM Conference on Recommender Systems (RecSys '18)},
  year      = {2018},
  pages     = {224--232},
  publisher = {ACM},
  doi       = {10.1145/3240323.3240370},
  note      = {arXiv:1710.11214}
}

@inproceedings{mansoury2020feedback,
  author    = {Mansoury, Masoud and Abdollahpouri, Himan and Pechenizkiy, Mykola and Mobasher, Bamshad and Burke, Robin},
  title     = {Feedback Loop and Bias Amplification in Recommender Systems},
  booktitle = {Proceedings of the 29th ACM International Conference on Information \& Knowledge Management (CIKM '20)},
  year      = {2020},
  pages     = {2145--2148},
  publisher = {ACM},
  doi       = {10.1145/3340531.3412152},
  note      = {arXiv:2007.13019}
}

@inproceedings{jing2022dimensionalcollapse,
  author    = {Jing, Li and Vincent, Pascal and LeCun, Yann and Tian, Yuandong},
  title     = {Understanding Dimensional Collapse in Contrastive Self-Supervised Learning},
  booktitle = {International Conference on Learning Representations (ICLR)},
  year      = {2022},
  note      = {OpenReview YevsQ05DEN7; arXiv:2110.09348}
}

@inproceedings{chen2024ncl,
  author    = {Chen, Huiyuan and Lai, Vivian and Jin, Hongye and Jiang, Zhimeng and Das, Mahashweta and Hu, Xia},
  title     = {Towards Mitigating Dimensional Collapse of Representations in Collaborative Filtering},
  booktitle = {Proceedings of the 17th ACM International Conference on Web Search and Data Mining (WSDM '24)},
  year      = {2024},
  pages     = {106--115},
  publisher = {ACM},
  address   = {Merida, Mexico},
  doi       = {10.1145/3616855.3635832},
  note      = {arXiv:2312.17468}
}

@inproceedings{loveland2025weightdecay,
  author    = {Loveland, Donald and Ju, Mingxuan and Zhao, Tong and Shah, Neil and Koutra, Danai},
  title     = {On the Role of Weight Decay in Collaborative Filtering: A Popularity Perspective},
  booktitle = {Proceedings of the 31st ACM SIGKDD Conference on Knowledge Discovery and Data Mining V.2 (KDD '25)},
  year      = {2025},
  pages     = {1975--1986},
  publisher = {ACM},
  address   = {Toronto, ON, Canada},
  doi       = {10.1145/3711896.3737068},
  note      = {arXiv:2505.11318}
}

@inproceedings{zielnicki2025orthogonal,
  author    = {Zielnicki, Kevin and Hsiao, Ko-Jen},
  title     = {Orthogonal Low Rank Embedding Stabilization},
  booktitle = {Proceedings of the Nineteenth ACM Conference on Recommender Systems (RecSys '25)},
  year      = {2025},
  publisher = {ACM},
  doi       = {10.1145/3705328.3748141},
  note      = {arXiv:2508.07574}
}

@inproceedings{rendle2009bpr,
  author    = {Rendle, Steffen and Freudenthaler, Christoph and Gantner, Zeno and Schmidt-Thieme, Lars},
  title     = {{BPR}: Bayesian Personalized Ranking from Implicit Feedback},
  booktitle = {Proceedings of the Twenty-Fifth Conference on Uncertainty in Artificial Intelligence (UAI '09)},
  year      = {2009},
  pages     = {452--461},
  publisher = {AUAI Press},
  address   = {Montreal, QC, Canada},
  note      = {No DOI; extended version arXiv:1205.2618}
}

@article{bottou2018optimization,
  author  = {Bottou, L{\'e}on and Curtis, Frank E. and Nocedal, Jorge},
  title   = {Optimization Methods for Large-Scale Machine Learning},
  journal = {SIAM Review},
  volume  = {60},
  number  = {2},
  pages   = {223--311},
  year    = {2018},
  doi     = {10.1137/16M1080173}
}

\appendix
\section{Synthetic Validation of the Phase Boundary}
\label{app:sim}

\paragraph{Environment.} We build a setting where every assumption
of Theorem~\ref{thm:contraction} holds exactly. There are $n=2000$ users,
$m=500$ items, and $d=16$. Item embeddings are drawn once from a
zero-mean Gaussian with an anisotropic covariance whose eigenvalues
are log-spaced from $0.001$ to $0.1$, then frozen for the entire run
(Assumption~\ref{ass:stationary}). Interaction counts follow a power law and positives
are drawn from the popularity-weighted distribution with $\alpha=1.0$
for every user, so every user is mainstream. Negatives
are uniform with $\lneg=1$. The optimizer is plain
SGD with constant learning rate $\eta=0.001$ (Assumption~\ref{ass:bpr}). We
measure $V_C^{0}=-112.35$ from the initialization and
$\bar\delta = 0.5000$ at warmup, and we use $\varepsilon=10^{-9}$. The
same seed and initialization are shared across all arms, so only
$\gamma$ differs. The left-hand side of the corollary is
$\log\det Q = -65.75$, computed once from the known items and weights
and identical across arms. That invariance is also the check that caught
configuration errors during development, since any arm reporting a
different $\log\det Q$ has a bug rather than a finding.

\paragraph{Results.} Table~\ref{tab:sim} reports the grid. The
closed-form crossover is $\gamma^{*}=2.30\times10^{-3}$. At all four
grid points above $\gamma^{*}$ the predicted sign is correct. At
$\gamma=0.1$ the predicted plateau is
$V_C^{*} = d\log(\eta\bar\delta^{2}/2\gamma) + \log\det Q = -172.7$
against an observed plateau of $-174.4$, a gap of $1.7$ nats, about
$0.11$ nats per dimension. Regressing $\log\lambda_k(C^{*})$ on
$\log\mu_k(Q)$ at this plateau yields slope $0.926$ with
$R^{2}=0.998$, so the sub-unit spectral compression observed on real
data under Adam also appears under SGD.

\begin{table}[h]
\caption{Boundary predictions against observed $\Delta V_C$ in the
synthetic environment. LHS $= \log\det Q = -65.75$ throughout.}
\label{tab:sim}
\centering
\small
\begin{tabular}{lrrllr}
\toprule
$\gamma$ & RHS & margin & pred. & obs. $\Delta V_C$ & steps \\
\midrule
$10^{-4}$        & $-115.92$ & $-50.17$ & exp. & $-8.61\%$  & $10^{6}$ \\
$3\times10^{-4}$ & $-98.34$  & $-32.59$ & exp. & $-10.18\%$ & $10^{6}$ \\
$10^{-3}$        & $-79.08$  & $-13.33$ & exp. & $-14.62\%$ & $10^{6}$ \\
$3\times10^{-3}$ & $-61.50$  & $+4.25$  & con. & $-19.04\%$ & $4\times10^{5}$ \\
$10^{-2}$        & $-42.23$  & $+23.52$ & con. & $-29.88\%$ & $2.7\times10^{5}$ \\
$3\times10^{-2}$ & $-24.66$  & $+41.09$ & con. & $-41.15\%$ & $1.2\times10^{5}$ \\
$10^{-1}$        & $-5.39$   & $+60.36$ & con. & $-55.26\%$ & $4.3\times10^{4}$ \\
\bottomrule
\end{tabular}
\end{table}

\paragraph{Below the crossover.} The three arms below $\gamma^{*}$
predict expansion but contract persistently, still declining at one
million steps, which is two full relaxation times $1/(2\eta\gamma)$ at
$\gamma=10^{-3}$. The decayed $\bar\delta$ cannot account for this. At
$\gamma=10^{-4}$ the fixed point computed with the decayed
$\bar\delta=0.383$ sits at $-70.7$, far above the observed $-122.0$,
and the system moves away from it. The cause is the state dependence
of $\delta_{ij}$ itself, which is anticorrelated with $\mathbf{u}$.
Users who already fit an item direction receive smaller updates along
it, an extra contractive force the constant-$\bar\delta$
approximation leaves out. Consistent with this, $\bar\delta$ decays $20$
to $23$ percent in the three fitting arms against roughly $3$ percent
at $\gamma=0.1$, where weight decay keeps scores near zero and the
approximation is nearly exact.

\paragraph{What this scopes.} The two sides of the boundary are not
equally reliable, and the simulation is how we know. Above the
crossover the prediction holds in sign at every grid point and in
magnitude to $0.11$ nats per dimension, under the theorem's own
optimizer, with every assumption true by construction. Below it the
sign prediction fails in a specific and explainable way. Because the
failure appears here, in an environment with no confounds, it isn't an
artifact of real data, of Adam, or of the mainstream-user restriction.
It is a limitation of the constant-$\bar\delta$ approximation in the
regime where the model actually fits. A reader deciding whether to
trust the boundary on their own system should read it as a reliable
detector of the strong-collapse regime and not as a guarantee of
expansion below it.

\end{document}